\documentclass[10pt,envcountsame]{llncs}
\usepackage{amsmath,amssymb,array}
\usepackage{verbatim}
\usepackage[hidelinks]{hyperref}
\usepackage{a4wide}

\usepackage{tikz}
\usetikzlibrary{arrows,snakes,matrix,backgrounds}
\usetikzlibrary{fit,positioning}

\mathchardef\mhyphen="2D
\usepackage[algoruled,linesnumbered,noend]{algorithm2e}
\usepackage[utf8]{inputenc} 
\usepackage{todonotes}

\renewcommand\leq\leqslant
\renewcommand\geq\geqslant

\mathchardef\mhyphen="2D

\newcommand{\MinCov}[1]{\ensuremath{\mathsf{Min~#1 \mhyphen Club~Cover}}}
\newcommand{\Cover}[2]{\ensuremath{\mathsf{#1 \mhyphen Club~Cover(#2)}}}
\newcommand{\Cov}[1]{\ensuremath{\mathsf{#1 \mhyphen Club~Cover}}}

\newcommand{\DoubleSat}{\ensuremath{\mathsf{5 \mhyphen Double \mhyphen Sat}}}
\newcommand{\Sat}{\ensuremath{\mathsf{Sat}}}
\newcommand{\TSat}{\ensuremath{\mathsf{3 \mhyphen Sat}}}

\newcommand{\MinPart}{\ensuremath{\mathsf{Minimum~Clique~Partition}}}
\newcommand{\PartT}{\ensuremath{\mathsf{3 \mhyphen Clique~Partition}}}
\begin{document}

\title{Covering with Clubs: Complexity and Approximability}
\author{Riccardo Dondi\inst{1}
\and Giancarlo Mauri \inst{2}
\and Florian Sikora \inst{3}
\and Italo Zoppis \inst{2}}
\institute{
%Dipartimento di Lettere, Filosofia, Comunicazione\\
Universit\`a degli Studi di Bergamo, Bergamo, Italy \email{riccardo.dondi@unibg.it}\\
\and
%DISCo, 
Universit\`a degli Studi di Milano-Bicocca, Milano - Italy \email{ \{mauri,zoppis\}@disco.unimib.it}\\
\and
Universit\'{e} Paris-Dauphine, PSL Research University, CNRS UMR 7243, LAMSADE, 75016 Paris, France \email{florian.sikora@dauphine.fr}
}
\maketitle

\begin{abstract}
Finding cohesive subgraphs in a network is a well-known problem in graph theory. 
Several alternative formulations of cohesive subgraph have been proposed,
a notable example being %clique.
%One of them is 
$s$-club, 
which is a subgraph where each vertex is at distance at most $s$ 
to the others. 
%Here we consider $s$-clubs, %that have been proposed has an alternative to clique,
Here we consider %a natural variant of the well-known
%\ensuremath{\mathsf{Minimum~Clique~Partition}} problem, 
%where we aim to 
the problem of covering a given graph with the minimum number of $s$-clubs.
%, instead of cliques. 
We study the computational and approximation complexity of this problem, 
when $s$ is equal to 2 or 3. % (when $s=1$, it corresponds to Clique Partition).
First, we show that deciding if there exists a cover 
of a graph with three $2$-clubs is NP-complete, and that 
deciding if there exists a cover 
of a graph with two $3$-clubs is NP-complete.
Then, we consider the approximation complexity 
of covering a graph with the minimum number of $2$-clubs 
and $3$-clubs. 
We show that, given a graph $G=(V,E)$ to be covered, 
covering $G$ with the minimum number of $2$-clubs is
not approximable  within factor $O(|V|^{1/2 -\varepsilon})$, 
for any $\varepsilon>0$, and 
covering $G$ with the minimum number of $3$-clubs is
not approximable  within factor $O(|V|^{1 -\varepsilon})$, for any $\varepsilon>0$.
On the positive side, we give an approximation
algorithm of factor $2|V|^{1/2}\log^{3/2} |V|$ for covering 
a graph with the minimum number of $2$-clubs.
\end{abstract}

\section{Introduction}
\label{sec:introduction}

The quest for modules inside a network is a well-known and deeply studied problem in network analysis, 
with several application in different fields, like computational biology or social network analysis.
A highly investigated problem is that of finding cohesive subgroups inside a network %(for example see~\cite{DBLP:conf/kdd/SozioG10}), 
which in graph theory translates 
in highly connected subgraphs. 
A common approach is to look for cliques (i.e. complete graphs),
and several combinatorial problems have been considered,
notable examples being the 
{\sf Maximum Clique problem} (\cite[GT19]{garey}), 
the {\sf Minimum Clique Cover} problem (\cite[GT17]{garey}),
and the {\sf Minimum Clique Partition} problem (\cite[GT15]{garey}).
This last is a  classical problem in theoretical computer science, whose goal is to partition the vertices of a graph into the minimum number of cliques.
The {\sf Minimum Clique Partition} problem has been deeply studied since the seminal paper of Karp~\cite{DBLP:conf/coco/Karp72}, studying its complexity in several graph classes ~\cite{DBLP:journals/dam/CerioliFFMPR08,DBLP:journals/ita/CerioliFFP11,DBLP:journals/algorithmica/PirwaniS12,DBLP:journals/gc/DumitrescuP11}.
%
%like cubic graphs~\cite{DBLP:journals/dam/CerioliFFMPR08},
%unit-disk 
%graphs~\cite{DBLP:journals/ita/CerioliFFP11,DBLP:journals/algorithmica/PirwaniS12,DBL%P:journals/gc/DumitrescuP11}.
%and clique-width bounded graphs~\cite{DBLP:conf/wg/EspelageGW01}.
%In the \ensuremath{\mathsf{Minimum~(Edge)~Clique~Cover}} problem, the goal is to find %the minimum number of cliques such that each edge belongs to at least one clique. 
%It is hard to approximate in polynomial-time~\cite{Ausiello-Crescenzi}.

In some cases, asking for a complete subgraph is 
too restrictive, as interesting highly connected graphs
may have some missing edges due to noise in the data considered
or because some pair may not be directly connected by an edge in the subgraph
of interest. 
To overcome this limitation of the clique approach, 
alternative definitions of highly connected graphs have been proposed,
leading to the concept of 
\emph{relaxed clique}~\cite{DBLP:journals/algorithms/Komusiewicz16}.
A relaxed clique is a graph $G=(V,E)$ whose vertices satisfy a property which
is a relaxation of the clique property. 
Indeed, a clique is a subgraph 
%having density $\frac{|V|(|V|-1)}{2}$ \todo{F: not sure if this is the correct definition of density. I would call this the number of edges (and the density of a graph is the number of edges divided by the maximum number of edges, so that clique got density 1)}, 
whose vertices are all at distance one 
from each other and have the same degree (the size of the clique minus one). 
Different definitions of relaxed clique are obtained by modifying one of the properties of
clique, thus leading to distance-based relaxed cliques, degree-based relaxed cliques, and so on 
(see for example~\cite{DBLP:journals/algorithms/Komusiewicz16}).

%\todo{try to sell more the huge literature on both problems (clique cover and s-club detection)}
%\todo{Riccardo: I have added a part on clique cover, for s-clubs I think we have many references later}

%\todo{Covering with stars is also a known problem (which is a special case of covering with 2-clubs)}
%\todo{Riccardo: I'm not sure; do you think we should add some references about it?}

In this paper, we focus on a distance-based relaxation. 
In a clique all the vertices are required to be at distance at most one
from each other. 
Here this constraint is relaxed, so that the vertices have to be at distance at most $s$, 
for an integer $s \geq 1$. A subgraph whose vertices are all distance
at most $s$ is called an \emph{$s$-club} (notice that, when $s=1$, an $s$-club is exactly a clique).
The identification of $s$-clubs inside a network has been applied to social networks~\cite{Mokken79,SociometricClique,DBLP:journals/snam/LaanMM16,DBLP:journals/snam/MokkenHL16,DBLP:conf/biostec/ZoppisDSCSM18}, 
and biological 
networks~\cite{DBLP:journals/jco/BalasundaramBT05}.
%networks~\cite{DBLP:conf/evoW/Pasupuleti08,DBLP:journals/jco/BalasundaramBT05}.
Interesting recent studies have shown the relevance of finding $s$-clubs 
in a network~\cite{DBLP:journals/snam/LaanMM16,DBLP:journals/snam/MokkenHL16}, 
in particular focusing on finding $2$-clubs in real networks like DBLP or 
a European corporate network.

Contributions to the study of $s$-clubs mainly focus on the 
{\sf Maximum s-Club} problem, 
that is the problem of finding an $s$-club of maximum size.
%(the problem is exactly the Maximum Clique problem).
%%aggiungere citazione
%The complexity of the problem has been extensively studied. 
{\sf Maximum s-Club} is known to be NP-hard, 
for each $s\geq 1$~\cite{DBLP:journals/eor/BourjollyLP02}.
Even deciding whether there exists an $s$-club larger than a given size 
in a graph of diameter $s+1$ is NP-complete, 
for each $s\geq 1$~\cite{DBLP:journals/jco/BalasundaramBT05}.
The {\sf Maximum s-Club} problem has been studied also in the approximability and parameterized complexity framework.
A polynomial-time approximation algorithm with factor $|V|^{1/2}$ for every $s \geq 2$ on an input graph $G=(V,E)$ has been designed~\cite{Asahiro2017}. 
This is optimal, since the problem is not approximable within factor $|V|^{1/2 - \varepsilon}$, 
on an input graph $G=(V,E)$, for each $\varepsilon >0 $ and $s \geq 2$~\cite{Asahiro2017}.
%On the positive side, polynomial-time approximation algorithms with factor $|V|^{1/2}$ for every even $s\geq 2$, and factor $|V|^{2/3}$ for every odd $s \geq 3$ have been designed~\cite{DBLP:conf/latin/AsahiroMS10}. 
As for the parameterized complexity framework, the problem is known to be fixed-parameter tractable, when parameterized by the size 
of an $s$-club ~\cite{DBLP:journals/ol/SchaferKMN12,DBLP:journals/dam/KomusiewiczS15,DBLP:journals/computing/ChangHLS13}.
The {\sf Maximum s-Club} problem has been investigated also 
for structural parameters and specific graph classes~\cite{DBLP:journals/jgaa/HartungKN15,DBLP:journals/dam/GolovachHKR14}.

In this paper, we consider a different combinatorial problem, 
where we aim at covering the vertices of a network
with a set of subgraphs. 
Similar to {\sf Minimum Clique Partition}, we consider the problem of covering a graph with the minimum number of $s$-clubs such that each vertex belongs to an $s$-club.
We denote this problem by $\MinCov{s}$,
%Min $s$-Club Cover, 
and we focus in particular on the cases $s=2$ and $s=3$. 
We show some analogies and differences between %Min $s$-Club Cover
$\MinCov{s}$ 
and {\sf Minimum Clique Partition}.
%the two problems.
%that the problem differs slightly 
%from Minimum Clique Cover in several aspects.
We start in Section~\ref{sec:complexity} by considering 
the computational complexity of the problem of covering 
a graph with two or three $s$-clubs. 
%(notice that when we ask for the existence of a single $s$-club
%that covers a graph, we have to simply check in polynomial-time
%if the given graph is an $s$-club).
This is motivated by the fact that 
{\sf Clique Partition} is known to be in P when
we ask whether there exists a partition of the 
graph consisting of two cliques, 
while it is NP-hard to decide whether there exists a partition of the 
graph consisting of three cliques~\cite{DBLP:journals/tcs/GareyJS76}.
%(since 
%\ensuremath{\mathsf{Clique~Partition}} is equivalent to Graph Coloring on the %complementary graph).
As for {\sf Clique Partition}, we show that %$2$-Club Cover 
%\todo{This is a bit sloppy to call this problem "min" something when we want to decide for 2 (but the decision name is given only afterwards)} 
%\todo{Riccardo: it's true; I changed this part a bit, I hope it works now.} 
it is NP-complete to decide whether there exist three $2$-clubs that cover a graph. 
On the other hand, we show that, unlike {\sf Clique Partition}, %Min $3$-Club Cover 
it is NP-complete to decide whether there exist two $3$-clubs that cover a graph.
These two results imply also that $\MinCov{2}$ and $\MinCov{3}$
%Min $2$-Club Cover and Min $3$-Club Cover 
do not belong to the class XP for the parameter "number of clubs" in a cover.

%\todo[inline]{It may be worth to look at the edge cover version, where the size of the solution is bigger and where the original problem is FPT.}
%\todo[inline]{Riccardo: this is interesting, we should have a look, for another paper :)}
%\todo[inline]{I suspect the vertex partition version to be similar to the clique version by taking $G^s$. Actually, the vertex version should be similar by looking at a "colored" edge version where only original edges are mandatory}
%\todo[inline]{Riccardo: I don't know if I get it... it's true we can use two colors, but 
%is it really useful? The same approach could be used for relating max-clique and max $s$-club}

Then, we consider the approximation complexity of $\MinCov{2}$ and $\MinCov{3}$.  %covering a network with the minimum number %of $2$-clubs and $3$-clubs. 
We recall that, given an input graph $G=(V,E)$, 
{\sf Minimum Clique Partition} is not approximable within factor $O(|V|^{1-\varepsilon})$, for any $\varepsilon > 0$, unless $P=NP$~\cite{DBLP:journals/toc/Zuckerman07}.
Here we show that 
$\MinCov{2}$ %nd $\MinCov{3}$
%Min $2$-Club Cover and Min $3$-Club Cover
has a slightly different behavior, while $\MinCov{3}$ is similar to \ensuremath{\mathsf{Clique~Partition}}.
Indeed, in Section~\ref{sec:HardApprox} we prove that $\MinCov{2}$ is not approximable within factor $O(|V|^{1/2 -\varepsilon})$,
for any $\varepsilon>0$,
unless $P=NP$ ,
while $\MinCov{3}$ is not approximable within factor $O(|V|^{1 -\varepsilon})$,
for any $\varepsilon>0$,  unless
$P=NP$. 
%(see Section~\ref{sec:HardApproxMinCov2} and Section~\ref{sec:HardApproxMinCov3}), 
In Section~\ref{sec:ApproxAlgo}, we present a greedy approximation algorithm that has factor 
$2|V|^{1/2}\log^{3/2} |V|$ for $\MinCov{2}$, which almost match the inapproximability result for the problem.
%Min $2$-Club Cover, 
%and $O(|V|^{13/16} \log |V|)$, 
%for $\MinCov{3}$.
%Min $3$-Club Cover.
We start the paper by giving in Section~\ref{sec:preliminaries} some definitions and 
by formally defining the problem we are interested in.
%Some of the proofs (marked with $\bigstar$) are omitted due to space constraint.

\section{Preliminaries}
\label{sec:preliminaries}

%First, notice that in the paper we consider only undirected graphs. 
Given a graph $G=(V,E)$ and a subset $V' \subseteq V$, we denote by
$G[V']$ the subgraph of $G$ induced by $V'$.
Given two vertices $u, v \in V$, the distance between $u$ and $v$ in $G$, 
denoted by $d_G(u,v)$, is the length of a shortest path from $u$ to $v$.
The diameter of a graph $G=(V,E)$ is the maximum distance between two vertices of $V$.
Given a graph $G=(V,E)$ and a vertex $v \in V$, we denote by $N_G(v)$ the set of neighbors of $v$,
that is $N_G(v)= \{u: \{v,u\} \in E \}$. 
%We denote by $N(v)$, for $1 \leq q \leq |V|$, 
%the set of vertices of $G$ having distance $q$ from $v$; we denote by %$N^*_q(v)$, for $1 \leq q \leq |V|$, 
%the set of vertices having distance at most $q$ from $v$. 
We denote by $N_G[v]$ the close neighborhood of $V$, that is 
$N_G[v]= N_G(v) \cup \{v\}$.
Define $N_G^l(v)= \{u: \text{ $u$ has distance at most $l$ from $v$} \}$,
with $1 \leq l \leq 2$. 
Given a set of vertices $X \subseteq V$ and $l$, with 
$1 \leq l \leq 2$,
define $N^l_G(X)= \bigcup_{u \in X} N_G^l(u)$.
We may omit the subscript $G$ when it is clear from the context.
Now, we give the definition of $s$-club, which is fundamental for the paper.

\begin{definition}
\label{def:2-club}
Given a graph $G=(V,E)$, and a subset $V' \subseteq V$, 
$G[V']$ is an $s$-club if it has diameter at most $s$.
\end{definition}

Notice that an $s$-club must be a connected graph.
We present now the formal definition of the 
\ensuremath{\mathsf{Minimum~s\mhyphen Club~Cover}}
%Minimum $s$-Club Cover 
problem 
we are interested in.

\smallskip

%\begin{problem}
%\label{prob:MinSClubCover} 
\noindent\ensuremath{\mathsf{Minimum~s\mhyphen Club~Cover}}  ($\MinCov{s}$)\\
\textbf{Input:} a graph $G=(V,E)$ and an integer $s \geq 2$.\\
\textbf{Output:} a minimum cardinality collection $\mathcal{S}= \{ V_1, \dots, V_h \}$ 
such that, for each $i$ with $1 \leq i \leq h$, $V_i \subseteq V$,
$G[V_i]$ is an $s$-club, 
and, for each vertex $v \in V$, there exists a set $V_j$, 
with $1 \leq j \leq h$, such that $v \in V_j$.
%\end{problem}

We denote by $\Cover{s}{h}$, with $1 \leq h \leq |V|$, 
the decision version of $\MinCov{s}$ that asks
whether there exists a cover of $G$ consisting of at most $h$ $s$-clubs.

Notice that while in {\sf Minimum Clique Partition}
we can assume that
the cliques that cover a graph $G=(V,E)$
partition $V$, hence the cliques are vertex disjoint, 
we cannot make this assumption for $\MinCov{s}$. 
Indeed, in a solution of $\MinCov{s}$, 
a vertex may be covered by more than one $s$-club, in
order to have a cover consisting of the minimum number of $s$-clubs.
Consider the example of Fig.~\ref{fig:ExCover}. The
two $2$-clubs induced by $\{v_1,v_2,v_3,v_4,v_5 \}$
and $\{v_1,v_6,v_7,v_8,v_9 \}$ cover $G$, and both these $2$-clubs contain vertex $v_1$.
However, if we ask for a partition of $G$, we need at least three $2$-clubs.
%(for example the $2$-clubs induced by $\{v_1,v_2,v_3,v_4,v_5 \}$, $\{v_6,v_7 \}$ and 
%$\{v_8,v_9 \}$).
This difference between {\sf Minimum~Clique~Partition} and $\MinCov{s}$ 
is due to the fact that, while being a clique is a hereditary property, 
this is not the case for being an $s$-club. 
%Notice that being an $s$-club is a property which is not \emph{hereditary}.
If a graph $G$ is an $s$-club, then a subgraph of $G$ may not be an $s$-club 
(for example a star is a $2$-club, but the subgraph obtained by removing its center
is not anymore a $2$-club).
%\todo[inline]{This needs precisions I guess. Under the partition assumption, for the example, why couldn't we have $\{1,2,3,4,5\},\{6,7\},\{8,9\}$? }
%\todo[inline]{Riccardo: it is related to the minimum cover. I have added a sentence to explain it, let me know if it is clear.}

\begin{figure}
\centering
\begin{center}\def\svgwidth{0.31\textwidth}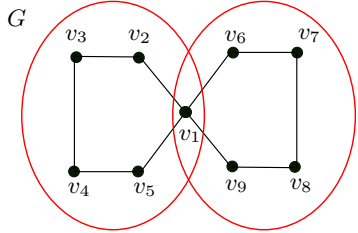\end{center}
\caption{A graph $G$ and a cover consisting of two $2$-clubs (induced by the 
vertices in the ovals). Notice that the $2$-clubs of this cover must both
contain vertex $v_1$. %If $v_1$ is contained only in one 
%$2$-club, for example in the $2$-club induced by 
%$\{ v_1, v_2, v_3, v_4, v_5 \}$, then two $2$-clubs are needed to cover 
%$\{ v_6, v_7, v_8, v_9 \}$, since the subgraph induced by 
%$\{ v_6, v_7, v_8, v_9 \}$ is not a $2$-club ($v_6$ and $v_9$ have
%distance $3$ in this subgraph).
}
\label{fig:ExCover}
\end{figure}

\section{Computational Complexity}
%\section{Computational Complexity of $\Cov{2}$ and $\Cov{3}$}
\label{sec:complexity}

In this section we investigate the computational complexity of $\Cov{2}$ and $\Cov{3}$ 
and we show that $\Cover{2}{3}$, that is deciding whether
there exists a cover of a graph $G$ with three $2$-clubs,
%is covering a graph with three $2$-clubs,
%is NP-complete 
and $\Cover{3}{2}$, 
that is deciding whether
there exists a cover of a graph $G$ with two $3$-clubs,
%that is covering a graph with two $3$-clubs, 
%is 
are 
NP-complete. 
%We start by showing that $\Cover{2}{3}$ is NP-complete.

\subsection{$\Cover{2}{3}$ is NP-complete}
\label{sec:compl:2-3}

In this section we show that $\Cover{2}{3}$ is NP-complete by giving a reduction from the $\PartT$ problem,
that is the problem of computing whether there exists a partition of a graph $G^p=(V^p,E^p)$ in three cliques.
Consider an instance $G^p=(V^p,E^p)$ of $\PartT$, we construct an instance 
$G=(V,E)$ of $\Cover{2}{3}$ (see Fig.~\ref{fig:3-club(2)}). 
The vertex set $V$ is defined as follows:
\[
V = \{ w_i: v_i \in V^p  \} \cup \{ w_{i,j}: \{ v_i,v_j \} \in E^p \wedge i < j \} \}
\]
The set $E$ of edges is defined as follows:
%\begin{dmath*}
%\sigma(2^{34}-1,2^{35},1)
%=-3+(2^{34}-1)/2^{35}+2^{35}\!/(2^{34}-1)
%+7/2^{35}(2^{34}-1)
%-\sigma(2^{35},2^{34}-1,1) -\sigma(2^2+22,11.2,2^2)
%\end{dmath*}.
%\[
%\begin{dmath*}
\begin{equation*}
%\begin{aligned}
\begin{split}
E = \{ \{w_i, w_{i,j} \}, \{w_i, w_{h,i} \}:  v_i \in V^p, w_i, w_{i,j}, w_{h,i} \in V \} 
\cup \\
\{ \{w_{i,j},w_{i,l}\}, \{w_{i,j},w_{h,i}\}, 
\{w_{h,i},
w_{z,i}\}:
w_{i,j}, w_{i,l}, w_{h,i}, w_{z,i} \in V \}
\end{split}
%\end{aligned}
\end{equation*}
%\end{dmath*}
%\wedge (i=h \lor i=l \lor j=h \lor j=l)\}
%\cup \{ \{w_{i,j},w_{l,j}\}: w_{i,j}, w_{l,j} \in V \}
%\]
Before giving the main results of this section, we prove a property of $G$.
%%TO BE DONE: CHANGE THE GRAPH
\begin{lemma}
\label{lem:cover2-3-Prop1}
Let $G^p=(V^p,E^p)$ be an instance of $\PartT$ and let $G=(V,E)$ be the corresponding
instance of $\Cover{2}{3}$. Then, given two vertices $v_i, v_j \in V^p$ and the corresponding 
vertices $w_i, w_j \in V$:
\begin{itemize}
\item if $\{ v_i,v_j \} \in E^p$, then $d_G(w_i,w_j) = 2$
\item if $\{ v_i,v_j \} \notin E^p$, then $d_G(w_i,w_j) \geq 3$
\end{itemize} 
\end{lemma}
\begin{proof}
%By construction there is no vertex $w_{i,j} \in E$. 
Notice that $N_G(w_i)= \{w_{i,z}: \{ v_i,v_z \} \in E^p \wedge i<z \} \cup 
\{w_{h,i}: \{ v_i,v_h \} \in E^p \wedge h<i \}$.
It follows that $w_j \in N^2_G(w_i)$ if and only if there exists a vertex 
$w_{i,j}$ (or $w_{j,i}$), which is adjacent
to both $w_i$ and $w_j$. But then, by construction, 
$w_j \in N^2_G(w_i)$ if and only if $\{ v_i,v_j \} \in E^p$.
\qed\end{proof}

We are now able to prove the main properties of the reduction.

\begin{lemma}
\label{lem:cover2-3-V1}
Let $G^p=(V^p,E^p)$ be a graph input of $\PartT$ and let $G=(V,E)$ be the corresponding
instance of $\Cover{2}{3}$. Then, given a solution of $\PartT$ on $G^p=(V^p,E^p)$, we can 
compute in polynomial time a solution of $\Cover{2}{3}$ on $G=(V,E)$.
\end{lemma}
\begin{proof}
Consider a solution of $\PartT$ on $G^p=(V^p,E^p)$, and let $V^p_{1}$, $V^p_{2}$, $V^p_{3} \subseteq V^p$ 
be the sets of vertices of $G^p$ that partition $V^p$.
We define a solution 
of $\Cover{2}{3}$ on $G=(V,E)$ as follows. For each $d$, with $1 \leq d \leq 3$,
define 
\[
V_d = \{ w_j \in V: v_j \in V^p_{d} \} \cup \{ w_{i,j}:  v_i \in V^p_{d} \}
\]

We show that each $G[V_d]$, with $1 \leq d \leq 3$, is a $2$-club. 
Consider two vertices $w_i, w_j \in V_d$, with $1 \leq i < j \leq |V|$. 
Since they correspond to two vertices 
$v_i, v_j \in V^p$ that belong to a clique of $G^p$,
it follows that $\{ v_i,v_j \} \in E^p$ and $w_{i,j} \in V_d$. 
Thus $d_{G[V_d]}(w_i,w_j) = 2$.
Now, consider the vertices $w_i \in V_d$,  with $1 \leq i \leq |V|$, 
and $w_{h,z} \in V_d$, with $1 \leq h < z \leq |V|$. 
If $i = h$ or $i=z$, assume w.l.o.g. $i=h$,
then  by construction $d_{G[V_d]}(w_i,w_{i,z}) = 1$.
Assume that $i \neq h$ and $i \neq z$ (assume w.l.o.g. that $i <h <z$), since $w_{h,z} \in V_d$,
it follows that $w_h \in V_d$. Since $w_i, w_h \in V_d$, it follows that 
$w_{i,h} \in V_d$. By construction, there exist edges 
$\{ w_{i,h}, w_{h,z} \}$, $\{ w_i, w_{i,h} \}$ in $E^p$, thus implying that
$d_{G[V_d]}(w_i,w_{h,z}) = 2$.
Finally, consider two vertices $w_{i,j}, w_{h,z} \in V_d$,  with $1 \leq i<j \leq |V|$
and  $1 \leq h < z \leq |V|$.
Then, by construction, $w_i \in V_d$ and 
$w_h \in V_d$. But then, 
$w_{i,h}$ belongs to $V_d$, and, by construction, $\{w_{i,j},w_{i,h} \} \in E$
and $\{w_{h,z},w_{i,h} \} \in E$. It follows that 
$d_{G[V_d]}(w_{i,j},w_{h,z}) = 2$.

We conclude the proof observing that, by construction, since $V^p_1, V^p_2, V^p_3$ partition $V^p$,
it holds that $V = V_1 \cup V_2 \cup V_3$, thus $G[V_1]$, $G[V_2]$, $G[V_3]$ covers $G$.
\qed\end{proof}
\begin{figure}
\centering
\begin{center}\def\svgwidth{0.55\textwidth}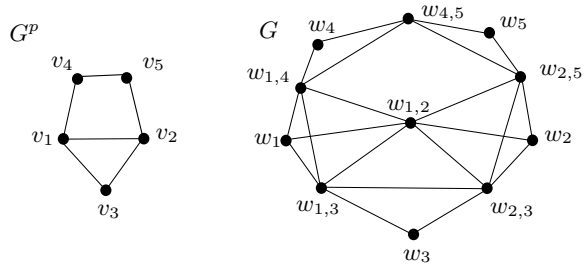\end{center}
%\def\svgwidth{\columnwidth}
%\svg[0.75]{2-club_3_tex}
\caption{An example of a graph $G^p$ input of $\PartT$ and the corresponding graph $G$ input of $\Cover{2}{3}$.}
\label{fig:3-club(2)}
\end{figure}
Based on Lemma~\ref{lem:cover2-3-Prop1}, we can prove the following result.
\begin{lemma}
\label{lem:cover2-3-V2}
Let $G^p=(V^p,E^p)$ be a graph input of $\PartT$ and let $G=(V,E)$ be the
corresponding instance of $\Cover{2}{3}$. 
Then, given a solution of $\Cover{2}{3}$ on $G=(V,E)$, 
we can compute in polynomial time a solution of $\PartT$ on $G^p=(V^p,E^p)$.
\end{lemma}
\begin{proof}
Consider a solution of $\Cover{2}{3}$ on $G=(V,E)$ consisting
of three $2$-clubs $G[V_1]$, $G[v_2]$, $G[V_3]$. 
Consider a $2$-club $G[V_d]$, with $1 \leq d \leq 3$.
By Lemma~\ref{lem:cover2-3-Prop1}, it follows that, for each $w_i, w_j \in V_d$, $\{ v_i,v_j \} \in E$.
As a consequence, we can define three cliques $G^p[V^p_1]$, $G^p[V^p_2]$, $G^p[V^p_3]$ in $G^p$ as follows.
For each $d$, with $1 \leq d \leq 3$, $V^p_d$ is defined as:
\[
V^p_d=\{ v_i: w_i \in V_d \}
\]
Next, we show that $G[V^p_d]$, with $1 \leq d \leq 3$, is indeed a clique.
By Lemma~\ref{lem:cover2-3-Prop1} if $w_i, w_j \in V_d$
then it holds $\{ v_i,v_j \} \in E$, thus by construction
$\{v_i,v_j\} \in E^p$ and $G[V^p_d]$ is a clique in $G^p$. 
Moreover, since $V_1 \cup V_2 \cup V_3 = V$,
then $V^p_1 \cup V^p_2 \cup V^p_3 =V^p$. Notice that 
$V^p_1$, $V^p_2$, $V^p_3$ may not be disjoint, but, starting
from $V^p_1$, $V^p_2$, $V^p_3$, it is easy to compute in polynomial
time a partition of $G^p$ in three cliques.
\qed\end{proof}

Now, we can prove the main result of this section.

\begin{theorem}
\label{teo:cover2-3}
$\Cover{2}{3}$ is NP-complete.
\end{theorem}
\begin{proof}
By Lemma~\ref{lem:cover2-3-V1} and Lemma~\ref{lem:cover2-3-V2} and from the 
NP-hardness of $\PartT$~\cite{DBLP:conf/coco/Karp72}, it follows that
$\Cover{2}{3}$ is NP-hard. The membership to NP follows easily from the fact that,
given three $2$-clubs of $G$,
it can be checked in polynomial time whether they are $2$-clubs and 
cover all vertices of $G$.
\qed\end{proof}

%\todo[inline]{Maybe we can find special graph classes where it remains NP-hard. Tho, covering with 3 2-clubs make the graph very dense... Also, the reduction "densify" the graph, so even if we start from a planar or a bounded degree instance of 3-coloring, we lose all. But maybe we can have such planarity or degree condition for minimization version (as for the clique cover problem)}
%\todo[inline]{Riccardo: I think that for the minimization variant we should be able to obtain some complexity results, but for this variant it looks hard}

\subsection{$\Cover{3}{2}$ is NP-complete}

In this section we show that $\Cover{3}{2}$ is NP-complete by giving a reduction from 
a variant of $\Sat$ called  $\DoubleSat$.
%Recall that $\PartT$ is polynomial-time solvable 
%if we want to decide whether there exists  a partition of a graph with only 2 cliques.
%$\DoubleSat$ is a variant of $\Sat$. 
%Moreover, 
Recall that a literal is positive if it is
a non-negated variable, while it is negative if it is a negated variable.

Given a collection of clauses $\mathcal{C}= \{ C_1, \dots, C_p \}$ 
over the set of variables $X=\{ x_1, \dots, x_q \}$,
where each $C_i \in \mathcal{C}$, with $1 \leq i \leq p$, contains exactly five literals
and does not contain both a variable and its negation, 
$\DoubleSat$ asks for a truth assignment to 
the variables in $X$ such that each clause $C_i$, with $1 \leq i \leq p$, 
is \emph{double-satisfied}. A clause $C_i$ is double-satisfied
by a truth assignment $f$ to the variables $X$ if there exist 
a positive literal and a negative literal in $C_i$ 
that are both satisfied by $f$.
Notice that we assume that there exist at least one positive literal and at least one negative literal
in each clause $C_i$, with $1 \leq i \leq p$, otherwise
$C_i$ cannot be doubled-satisfied.
Moreover, we assume that each variable in an instance
of $\DoubleSat$ appears both as a positive literal and a negative literal
in the instance. 
Notice that if this is not the case, for example a variable appears only
as a positive literal, we can assign a true value to the variable, as defining
an assignment to false does not contribute to double-satisfy any clause.
First, we show that $\DoubleSat$ is NP-complete, which may be of independent
interest.
%, by giving a reduction from $\TSat$.
%\todo{Riccardo: Florian, do you think we should add a definition of SAT or 3-SAT?}
%\todo{I suggest to give the definition in the proof which is in appendix s.t. we don't lose space :)}
\begin{theorem}
\label{teo:5DoubleSatHard}
$\DoubleSat$ is NP-complete.
\end{theorem}
\begin{proof}
We reduce from $\TSat$, where given a set $X_3$ of variables and a set $\mathcal{C}_3$ of clauses, which are a disjunction of 3 literals (a variable or the negation of a variable), we want to find an assignment to the variables such that all clauses are satisfied.
Moreover, we assume that each clause in $\mathcal{C}_3$ does not contain a positive variable $x$  and its
negation $\overline{x}$, since such a clause is obviously satisfied by any assignment. The same property 
holds also for the instance of $\DoubleSat$ we construct.

Consider an instance $(X_3,\mathcal{C}_3)$ of $\TSat$, we construct an instance
$(X, \mathcal{C})$ of $\DoubleSat$ as follows.
Define $X= X_3 \cup X_N$, where $X_3 \cap X_N = \emptyset$ and $X_N$ is defined
as follows:
\[
X_N = \{ x_{C,i,1}, x_{C,i,2}:C_i \in \mathcal{C}_3    \}
\]
The set $\mathcal{C}$ of clauses is defined as follows:
\[
\mathcal{C}= \{ C_{i,1}, C_{i,2}: C_i \in \mathcal{C}_3 \}
\]
where $C_{i,1}$, $C_{i,2}$ are defined as follows.
Consider $C_i \in \mathcal{C}_3=(l_{i,1} \vee l_{i,2} \vee l_{i,3})$, 
%with $1 \leq i \leq |\mathcal{C}|$, 
where $l_{i,p}$, with $1 \leq p \leq 3$ is a literal,
that is a variable (a positive literal) or a negated variable (a negative literal), 
the two clauses $C_{i,1}$ and $C_{i,2}$ are defined as follows:
\begin{itemize}
\item $C_{i,1} = l_{i,1} \vee l_{i,2} \vee l_{i,3} \vee x_{C,i,1} \vee \overline{x_{C,i,2}}$
\item $C_{i,2} = l_{i,1} \vee l_{i,2} \vee l_{i,3} \vee \overline{x_{C,i,1}} \vee x_{C,i,2}$
\end{itemize}
%where $x_{C,i,1}$ and $x_{C,i,2}$ are two variables not in $X_3$.
%We denote by $X_N$ the set of variables $x_{C,i,p}$, with $1 \leq i \leq |\mathcal{C}|$
%and $1 \leq p \leq 2$.
We claim that $(X_3,\mathcal{C}_3)$ is satisfiable if and only if 
$(X, \mathcal{C})$ is double-satisfiable. 

Assume that $(X_3,\mathcal{C}_3)$ is satisfiable and let $f$ be an assignment
to the variables on $X$ that satisfies $\mathcal{C}_3$. 
Consider a clause $C_i$ in $\mathcal{C}_3$, with $1 \leq i \leq |\mathcal{C}_3|$. 
Since it is satisfied by $f$, it follows that there
exists a literal $l_{i,p}$ of $C_i$, with $1 \leq p \leq 3$, that is satisfied by $f$. 
Define an assignment $f'$ on $X$ that is identical to $f$
on $X_3$ and, if $l_{i,p}$ is positive, then 
assigns value false to both $x_{C,i,1}$ and $x_{C,i,2}$,
if $l_{i,p}$ is negative, then 
assigns value true to both $x_{C,i,1}$ and $x_{C,i,2}$.
%Assume without loss of generality that
%$l_{i,p}$ is positive. Then, by assigning value false to %both $x_{C,i,1}$ and $x_{C,i,2}$,
It follows that both $C_{i,1}$ and $C_{i,2}$ are double-satisfied by $f'$.

Assume that $(X,\mathcal{C})$ is double-satisfied by an assignment $f'$.
Consider two clauses $C_{i,1}$ and $C_{i,2}$, with $1 \leq i \leq |\mathcal{C}|$, 
that are double-satisfied by $f'$, we claim that there exists
at least one literal of $C_{i,1}$ and $C_{i,2}$ not in $X_N$ 
which is satisfied. Assume this is not the case, then, if $C_{i,1}$ is double-satisfied,
it follows that $x_{C,i,1}$ is true and $x_{C,i,2}$ is false, thus implying that $C_{i,2}$ 
is not double-satisfied. 
Then, an assignment $f$ that is identical to
$f'$ restricted to $X_3$ satisfies
each clause in $\mathcal{C}$. 

Now, since $\TSat$ is NP-complete~\cite{DBLP:conf/coco/Karp72}, it follows that $\DoubleSat$ 
is NP-hard. 
The membership to NP follows from the observation that, given an assignment
to the variables on $X$, we can check in polynomial-time whether each clause 
in $\mathcal{C}$ is double-satisfied or not.
\qed\end{proof}

Let us now give the construction of the reduction from $\DoubleSat$ 
to $\Cover{3}{2}$. 
Consider an instance of $\DoubleSat$ consisting of a set $\mathcal{C}$ of clauses 
$C_1, \dots, C_p$ over set $X=\{x_1, \dots, x_q\}$ of variables.
%We assume that there is no clause in $\mathcal{C}$ that contains a variable $x$ and its 
%negation $\overline{x}$ (we can show that this property holds in Theorem \ref{teo:5DoubleSatHard}).
We assume that it is not possible to double-satisfy all the clauses 
by setting at most two variables to true or to false 
(this can be easily checked in polynomial-time).

Before giving the details, we present an overview of the reduction.
Given an instance $(X,\mathcal{C})$ of $\DoubleSat$, for each positive literal $x_i$, with $1 \leq i \leq q$, 
we define vertices $x_{i,1}^T$, $x_{i,2}^T$ and for each negative literal
$\overline{x_i}$, with $1 \leq i \leq q$, 
we define a vertex $x_i^F$. Moreover, for each clause 
$C_j \in \mathcal{C}$,
with $1 \leq j \leq p$,
we define a vertex $v_{C,j}$.
We define other vertices to ensure that some vertices have distance not greater 
than three and to force the membership to one of the two $3$-clubs of the solution
(see Lemma~\ref{lem:3-club(2)Prel1}).
The construction implies that for each $i$ with $1 \leq i \leq q$, $x_{i,1}^T$
and $x_i^F$ belong to different $3$-clubs (see Lemma~\ref{lem:3-club(2)Prel});
this corresponds to a truth assignment to the variables in $X$.
Then, we are able to show that each vertex $v_{C,j}$ belongs to the same $3$-club
of a vertex $x_{i,1}^T$, with $1 \leq i \leq q$, and of a vertex $x_h^F$,
with $1 \leq h \leq q$, adjacent to $v_{C,j}$ (see Lemma \ref{lem:3-club(2)2}); 
these vertices correspond to a positive literal
$x_i$ and a negative literal $\overline{x_h}$, respectively, that are satisfied by a truth assignment, 
hence $C_j$ is double-satisfied.
%
%if it is possible to satisfy all the clauses by setting
%all the variables to true or all the variable to false, we construct
%a graph consisting of exactly two vertices and no edge.
%Notice that this is case is easy checkable in polynomial time.
%Assume that it is not possible to satisfy all the clauses by setting
%all the variables to true or all the variable to false.
%Moreover, we assume that at least two variables are assigned a true value. 

%\todo{Can we have an informal idea of the construction and reduction?}
Now, we give the details of the reduction.
Let $(X,\mathcal{C})$ be an instance of $\DoubleSat$,
we construct an instance $G=(V,E)$ of $\Cover{3}{2}$  as follows 
(see Fig.~\ref{fig:3club-2}). 
The vertex set $V$ is defined as follows:
\[
V =
\{ r, r', r_T, r'_T, r^*_T, r_F, r'_F \} \cup \{ x_{i,1}^T, x_{i,2}^T, x_{i}^F: x_i \in X  \} \cup 
\{v_{C,j}: C_j \in \mathcal{C} \} \cup \{y_1, y_2, y \} 
\]

The edge set $E$ is defined as follows:
\begin{equation*}
\begin{split}
E = \{ \{ r,r'\}, \{ \{ r', r_T\} , \{ r', r^*_T\}  \{r', r_F \} \}  
\cup \{ \{r_T, x_{i,1}^T\}: x_i \in X  \} \\
\cup \{ \{r_F, x_{i}^F\}: x_i \in X   \} \cup 
%\]
%\[
\{ \{r'_T, x_{i,1}^T\}: x_i \in X   \} \cup 
\{ \{r'_F, x_{i}^F\}: x_i \in X   \} \cup \\
 \{  \{ x_{i,1}^T, x_{i,2}^T\}: x_i \in X \} \cup 
\{ \{ r^*_T, x_{i,2}^T\}, \{ y_1, x_{i,2}^T\} :x_i \in X  \} \cup \\
%\]
%\[
\{  \{ x_{i,2}^T, x_j^F\}: x_i,x_j \in X, i \neq j \} \cup
\{ \{ x_{i,1}^T, v_{C,j} \}: x_i \in C_j \} \cup 
\{ \{ x_{i}^F, v_{C,j} \}: \overline{x_i} \in C_j \} 
\cup \\
\{  \{ v_{C,j},y \}: C_j \in \mathcal{C} \}  \cup  
%\]
%\[
\{ \{ y,y_2 \}, \{ y_1,y_2 \}, \{y_1,r'_T \}, \{y_1,r'_F \} \}
\end{split}
\end{equation*}

%\]
%
%: 1 \leq i\ leq p \} \cup \{ \{w_b, \bar{v_{C,i,1}}\}: 1 \leq i \leq p \} \cup
%\]
%\[
%\{ \{v_{C,i,1}, v_{C,i,2}\}: 1 \leq i\ leq p \} \cup \{ \{\bar{v_{C,i,1}}, \bar{v_{C,i,2}\}}: 1 \leq i \leq p \} \cup
%\]
%\[
%\{ \{v_{C,i,2},v_j\}: x_j \in C_i, 1 \leq i \leq p \} \cup \{ \{v_{C,i,2},v_j \}: \bar{x_j} \in C_i, 1 \leq i \leq p \}
%\]

We start by proving some properties of the graph $G$.

\begin{lemma}
\label{lem:3-club(2)Prel1}
Consider an instance $(\mathcal{C}, X)$ of $\DoubleSat$ and let $G=(V,E)$ be the corresponding
instance of $\Cover{3}{2}$. 
Then, (1) $d_G(r',y)>3$, (2) $d_G(r,y)>3$, (3) $d_G(r,v_{C,j})>3$, for each $j$ with $1 \leq j \leq p$,  and (4) $d_G(r,r'_F)>3$, $d_G(r,r'_T)>3$.
\end{lemma}
\begin{proof}
We start by proving (1). 
Notice that any path from $r'$ to $y$ must pass through $r_T$, $r^*_T$ or $r_F$. 
Each of $r_T$, $r^*_T$ or $r_F$ is adjacent to vertices $x_{i,1}^T$, $x_{i,2}^T$ and 
$x_i^F$, with $1 \leq i \leq q$ (in addition to $r'$), and none of these vertices is adjacent 
to $y$, thus concluding that $d_G(r',y)>3$. 
Moreover, observe that for each vertex $v_{C,j}$, with $1 \leq j \leq p$, 
there exists a vertex %
$x_{i,1}^T$, with $1 \leq i \leq q$, or $x_h^F$,
with $1 \leq h \leq q$, that is adjacent to $v_{C,j}$, with 
$1 \leq j \leq p$, %and that by construction, 
%there exists two vertices 
thus $d_G(r',v_{C_j})=3$, for each $j$ with $1 \leq j \leq p$.
As a consequence of (1), it follows that (2) holds, that is $d_G(r,y)>3$. 
Since $d_G(r',v_{C_j})=3$, for each $j$ with $1 \leq j \leq p$, it holds 
(3) $d_G(r,v_{C,j})>3$.

Finally, we prove (4). 
Notice that 
$N^2_G(r) = \{ r',  r^*_T, r_T, r_F \}$
and that none of the vertices in $N^2_G(r)$
%the vertices at distance two from $r$ are %$r^*_T$, $r_T$, $r_F$, and none of these %vertices 
is adjacent to $r'_F$ and $r'_T$, thus $d_G(r,r'_F)>3$.
\qed\end{proof}

%\begin{figure}
%\centering
%\includegraphics[scale=.65]{3-club_New.pdf}
%\svg[0.8]{3-club_New_tex}
%\caption{An example of graph constructed by the reduction 
%from $\DoubleSat$ to $\Cover{3}{2}$.
%Given an instance of $\DoubleSat$ consisting of 
%$C_1= x_1 \vee \overline{x_2} \vee x_3 \vee \overline{x_4} \vee x_5$ and 
%$C_2= \overline{x_1} \vee x_2 \vee \overline{x_3} \vee x_4 \vee \overline{x_5}$,
%$G$ is the corresponding graph built by the reduction.}
%\label{fig:3club}
%\end{figure}

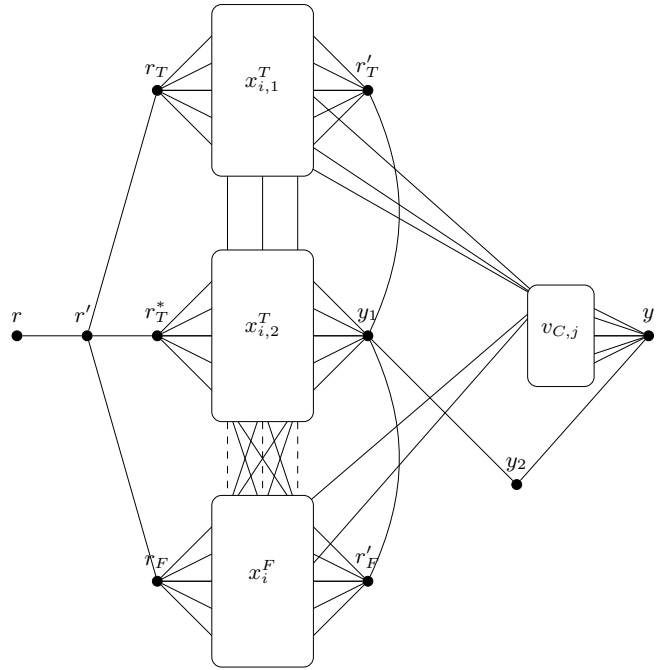
\begin{figure}
\centering
\resizebox{.6\textwidth}{!}
{
\begin{tikzpicture}
\tikzstyle{vertex}=[circle, draw, inner sep=0pt, minimum size=4pt, fill=black]
\node[vertex, label=$r$] (r) at (0,0) {};
\node[vertex, label=$r'~$, right of=r] (r') {};

\node[vertex, label=$r^*_T$, right of=r'] (rT) {};
\node[vertex, label=$r_T$, above of=rT, node distance=3.5cm] (r*) {};
\node[vertex, label=$r_F$, below of=rT, node distance=3.5cm] (rF) {};

\node[draw,right of=rT] (X1T1){};
\node[draw,below of=X1T1] (X1T2){};
\node[draw,above of=X1T1] (X1T3){};
\node[draw,right of=X1T3] (X1T4){};
\node[draw,right of=X1T2] (X1T5){};
\node[draw,right of=X1T1] (X1T6){};
\node[draw,right of=X1T3,node distance=0.5cm] (X1T7){};
\node[draw,right of=X1T2,node distance=0.5cm] (X1T8){};

\node[right of=r*] (X2T1){};
\node[below of=X2T1] (X2T2){};
\node[above of=X2T1] (X2T3){};
\node[right of=X2T3] (X2T4){};
\node[right of=X2T2] (X2T5){};
\node[right of=X2T1] (X2T6){};
\node[draw,right of=X2T3,node distance=0.5cm] (X2T7){};
\node[draw,right of=X2T2,node distance=0.5cm] (X2T8){};

\node[draw,right of=rF] (XF1){};
\node[draw,below of=XF1] (XF2){};
\node[draw,above of=XF1] (XF3){};
\node[draw,right of=XF3] (XF4){};
\node[draw,right of=XF2] (XF5){};
\node[draw,right of=XF1] (XF6){};
\node[draw,right of=XF3,node distance=0.5cm] (XF7){};
\node[draw,right of=XF2,node distance=0.5cm] (XF8){};

\node[draw,right of=X1T6,node distance=3.5cm] (C1){};
\node[draw,below of=C1,node distance=0.5cm] (C2){};
\node[draw,above of=C1,node distance=0.5cm] (C3){};
\node[draw,right of=C1,node distance=0.5cm] (C4){};
\node[draw,right of=C2,node distance=0.5cm] (C5){};
\node[draw,right of=C3,node distance=0.5cm] (C6){};

\node[vertex, label=$y_1$, right of=X1T6] (r'T) {};
\node[vertex, label=$r'_T$, right of=X2T6] (y1) {};
\node[vertex, label=$r'_F$, right of=XF6] (r'F) {};
\node[vertex, label=$y$, right of=C4] (y) {};

\draw (r) -- (r');
\draw (rT) -- (r');
\draw (r*) -- (r');
\draw (rF) -- (r');

\draw (rT) -- (X1T3) ;
\draw (rT) -- (X1T1) ;
\draw (rT) -- (X1T2) ;
\draw (rT) -- (X1T4) ;
\draw (rT) -- (X1T5) ;
\draw (rT) -- (X1T6) ;
\draw (r'T) -- (X1T3) ;
\draw (r'T) -- (X1T1) ;
\draw (r'T) -- (X1T2) ;
\draw (r'T) -- (X1T4) ;
\draw (r'T) -- (X1T5) ;
\draw (r'T) -- (X1T6) ;

\draw (r*) -- (X2T3);
\draw (r*) -- (X2T2);
\draw (r*) -- (X2T1);
\draw (r*) -- (X2T4);
\draw (r*) -- (X2T5);
\draw (r*) -- (X2T6);
\draw (y1) -- (X2T3);
\draw (y1) -- (X2T2);
\draw (y1) -- (X2T1);
\draw (y1) -- (X2T4);
\draw (y1) -- (X2T5);
\draw (y1) -- (X2T6);

\draw (rF) -- (XF3);
\draw (rF) -- (XF2);
\draw (rF) -- (XF1);
\draw (rF) -- (XF4);
\draw (rF) -- (XF5);
\draw (rF) -- (XF6);
\draw (r'F) -- (XF3);
\draw (r'F) -- (XF2);
\draw (r'F) -- (XF1);
\draw (r'F) -- (XF4);
\draw (r'F) -- (XF5);
\draw (r'F) -- (XF6);

\draw (X1T7) -- (X2T8);
\draw (X1T3) -- (X2T3);
\draw (X1T4) -- (X2T4);

\draw (C3) -- (X2T1);
\draw (C3) -- (X2T3);
\draw (C3) -- (XF4);
\draw (C3) -- (X2T5);
\draw (C3) -- (XF6);
%\draw (C1) -- (X1T3);
%\draw (C1) -- (XF4);
%
%\draw (C3) -- (X1T1);
%\draw (C3) -- (XF6);
%\draw (C3) -- (XF4);

\draw (XF3) edge[dashed] (X1T2);
\draw (XF3) -- (X1T8);
\draw (XF3) -- (X1T5);
\draw (XF7) -- (X1T2);
\draw (XF7) -- (X1T5);
\draw (XF4) -- (X1T2);
\draw (XF7) edge[dashed] (X1T8);
\draw (XF4) -- (X1T8);
\draw (XF4) edge[dashed] (X1T5);

\draw (y) -- (C1);
\draw (y) -- (C2);
\draw (y) -- (C3);
\draw (y) -- (C4);
\draw (y) -- (C5);
\draw (y) -- (C6);

\node[vertex, label=$y_2$, below right of=r'T, node distance=3cm] (y2) {};

\draw (r'T) edge[bend right=25] (y1);
\draw (r'T) -- (y2) -- (y);
\draw (r'F) edge[bend right=25] (r'T);

%\draw (X2T4) edge[bend left=140] (XF5);
%\draw (X2T4) .. controls (contr1) .. (XF5);

%,label={[name=lSp]left:$S'$}
\node[draw,fill=white,rectangle,rounded corners,fit= (X1T1) (X1T3) (X1T2) (X1T4)] (X1T) {$x_{i,2}^T$};
\node[draw,fill=white,rectangle,rounded corners,fit= (X2T1) (X2T3) (X2T2) (X2T4)] (X2T) {$x_{i,1}^T$};
\node[draw,fill=white,rectangle,rounded corners,fit= (XF1) (XF3) (XF2) (XF4)] (XF) {$x_{i}^F$};
\node[draw,fill=white,rectangle,rounded corners,fit= (C1) (C2) (C3) (C4) (C5) (C6)] (C) {$v_{C,j}$};
\end{tikzpicture}
}
\caption{Schematic construction for the reduction from $\DoubleSat$ to $\Cover{3}{2}$. % (dashed edges means there are no edge if the index is the same).
\label{fig:3club-2}
}
\end{figure}

Consider two sets $V_1 \subseteq V$ and $V_2 \subseteq V$, such that $G[V_1]$ and $G[V_2]$ 
are two $3$-clubs of $G$ that cover $G$.
As a consequence of Lemma~\ref{lem:3-club(2)Prel1}, 
it follows that $r$ and $r'$ are in exactly one of $G[V_1]$, $G[V_2]$, w.l.o.g. $G[V_1]$, %\todo{$y$ cannot be in the same club as
%$%r$}, 
while $r'_T$, $r'_F$, $y$ and $v_{C,j}$, for each $j$ with $1 \leq j \leq p$, belong 
to $G[V_2]$ and not to $G[V_1]$.
%the other $3$-club. 
%We assume that $r,r' \in V_1 \setminus V_2$, while $y, r'_T, r'_F \in V_2 \setminus V_1$,%\todo{contradiction with the previous sentence} 
%and $v_{C,j} \in V_2 \setminus V_1$, for each $j$ with $1 \leq j \leq p$.

Next, we show a crucial property of the graph $G$ built by the reduction.

\begin{lemma}
\label{lem:3-club(2)Prel}
Given an instance $(\mathcal{C}, X)$ of $\DoubleSat$, let $G=(V,E)$ be the corresponding
instance of $\Cover{3}{2}$. Then, for each $i$ with $1 \leq i \leq q$, 
$d_G(x_{i,1}^T,x_i^F)>3$.
\end{lemma}
\begin{proof}
Consider a path $\pi$ of minimum length that connects $x_{i,1}^T$ and $x_i^F$, 
with $1 \leq i \leq q$.
First, notice that, by construction, the path $\pi$ after $x_{i,1}^T$ must pass through one of 
these vertices: $r_T$, $r'_T$, $x_{i,2}^T$ or $v_{C,j}$, with $1 \leq j \leq p$.

We consider the first case, that is the path $\pi$ after $x_{i,1}^T$  passes through $r_T$. 
Now, the next vertex in $\pi$ is either $r'$ or $x_{h,1}^T$, with $1 \leq h \leq q$. 
Since  both $r'$ and $x_{h,1}^T$ are not adjacent to $x_i^F$, 
it follows that in this case the path $\pi$ has length greater than three.

We consider the second case, that is the path $\pi$ after $x_{i,1}^T$ passes through $r'_T$. 
Now, after $r'_T$, $\pi$ passes through either $y_1$ or $x_{h,1}^T$, with $1 \leq h \leq q$. 
Since both $y_1$ and $x_{h,1}^T$ are not adjacent to $x_i^F$, it follows that in this case
the path $\pi$ has length greater than three.

We consider the third case, that is the path after $x_{i,1}^T$  passes through $x_{i,2}^T$. 
Now, the next vertex of $\pi$  is either $r^*_T$ or $y_1$ or $x_h^F$, with $1 \leq h \leq q$
and $h \neq i$. 
Since $r^*_T$, $y_1$ and $x_{h}^F$ are not adjacent to $x_i^F$, it follows that in this case 
the path $\pi$ has length greater than three.

We consider the last case, that is the path after $x_{i,1}^T$ passes through $v_{C,j}$, 
with $1 \leq j \leq p$. 
We have assumed that $x_i$ and $\overline{x_i}$ do not belong to the same clause, thus by construction $x_{i}^F$ is not incident in $v_{C,j}$.
It follows that after $v_{C,j}$, the path $\pi$ must pass through either $y$ or 
$x_{h,1}^T$, with $1 \leq h \leq q$, or $x_{z}^F$, $1 \leq z \leq q$ and $z \neq i$.
Once again, since $y$, $x_{h,1}^T$ and $x_{z}^F$ are not adjacent to $x_i^F$, 
it follows that also in this case the path $\pi$ has length greater than three, thus concluding the proof.
\qed\end{proof}

Now, we are able to prove the main results of this section.

\begin{lemma}
\label{lem:3-club(2)1}
Given an instance $(\mathcal{C}, X)$ of $\DoubleSat$, let $G=(V,E)$ be the corresponding
instance of $\Cover{3}{2}$. 
Then, given a truth assignment that double-satisfies $\mathcal{C}$, 
we can compute in polynomial-time two $3$-clubs that cover $G$.
\end{lemma}
\begin{proof}
Consider a truth assignment $f$ on the set $X$ of variables that double-satisfies $\mathcal{C}$. 
In the following we construct two $3$-clubs $G[V_1]$ and $G[V_2]$ that cover $G$.
The two sets $V_1$, $V_2$ are defined as follows:

\[
V_1 = \{ r,r',r_T, r^*_T, r_F \} \cup \{ x_{i,1}^T, x_{i,2}^T:f(x_i)=false \} \cup 
\{ x_{i}^F, :f(x_i)= true \}
\]

\[
V_2 = \{r'_T, r'_F, y, y_1, y_2\} \cup 
\{ x_{i,1}^T, x_{i,2}^T:f(x_i) = true \} \cup \{ x_{i}^F:f(x_i)= false \cup \}
\]
\[ 
\{ v_{C,j}: 1 \leq j \leq p \}
\]

Next, we show that $G[V_1]$ and $G[V_2]$ are indeed two $3$-clubs that cover $G$.
First, notice that $V_1 \cup V_2 =V$, hence $G[V_1]$ and $G[V_2]$ cover $G$.
Next, we show that both $G[V_1]$ and $G[V_2]$ are indeed $3$-clubs. 

Let us first consider $G[V_1]$.
By construction, $d_{G[V_1]}(r,x_{i,1}^T) = 3$ and $d_{G[V_1]}(r,x_{i,2}^T) = 3$, for each $i$ with $1 \leq i \leq i \leq q$, and $d_{G[V_1]}(r,x_{i}^F) = 3$, for each $i$ with $1 \leq i \leq i \leq q$.
Moreover, $d_{G[V_1]}(r',x_{i,1}^T) = 2$ and $d_{G[V_1]}(r',x_{i,2}^T) = 2$, for each $i$ with $1 \leq i \leq q$, and $d_{G[V_1]}(r',x_{i}^F) = 2$, for each $i$ with $1 \leq i \leq i \leq q$.
As a consequence, it holds that $r_T$, $r'_T$ and $r_F$ have distance at most three in $G[V_1]$ from each vertex $x_{i,1}^T$, from each vertex $x_{i,2}^T$, and from each vertex $x_{i}^F$.
Since $r$, $r_T$, $r^*_T$ and $r_F$ are in $N(r')$, it follows that $r$, $r'$, $r_T$, $r^*_T$ and $r_F$ are at distance at most $2$ in $G[V_1]$.
Hence, we focus on vertices $x_{i,1}^T$, with $1 \leq i \leq q$, $x_{h,2}^T$, with $1 \leq h \leq q$ and $x_{j}^F$, with $1 \leq j \leq q$.
Since there exists a path that passes trough $x_{i,1}^T$, $r_T$, $x_{h,1}^T$ and $x_{h,2}^T$, vertices $x_{i,1}^T$, 
$x_{h,1}^T$ are at distance at most two in $G[V_1]$, while $x_{i,1}^T$, $x_{h,2}^T$ are at distance at most three 
in $G[V_1]$ (if $i=h$ they are at distance one).
Vertices $x_{h,2}^T$ and $x_{j}^F$ are at distance one in $G[V_1]$, since $h \neq j$ and $\{ x_{h,2}^T,x_j^F \} \in E$ 
by construction.
Finally, $x_{i,1}^T$ and $x_{j}^F$ are at distance two in $G[V_1]$, since there exists a path that passes trough $x_{i,1}^T$, $x_{i,2}^T$ and $x_{j}^F$ in $G[V_1]$, as $i \neq j$. 
It follows that $G[V_1]$ is a $3$-club.

We now consider $G[V_2]$. We recall that, for each $i$ with $1 \leq i \leq q$,
if $x_{i,1}^T$, $x_{i,2}^T \in V_2$, then $x_i^F \in V_1$.
Furthermore, we recall that we assume that each $x_i$ appears as
a positive and a negative literal in the instance of $\DoubleSat$, thus 
each vertex $x_{i,1}^T$, with $1 \leq i \leq q$, and each vertex 
$x_{h}^F$, with $1 \leq h \leq q$, are connected to some $V_{C,j}$, 
with $1 \leq j \leq p$.

First, notice that vertex $y$ is at distance at most three in $G[V_2]$ from each vertex of $V_2$, since it has distance one in $G[V_2]$ from each vertex $v_{C,j}$, with $1 \leq j \leq p$, thus distance two from $x_{i,1}^T$, with $1 \leq i \leq q$, and $x_{h}^F$, with $1 \leq h \leq q$, and three from $x_{i,2}^T$, with $1 \leq i \leq q$, $r'_T$ and $r'_F$.
Since $y$ is adjacent to $y_2$, it has distance one from $y_2$ and two from $y_1$.

Now, consider a vertex $v_{C,j}$, with $1 \leq j \leq p$. 
Since $f$ double-satisfies $\mathcal{C}$, it follows that there exist two vertices in $V_2$, $x_{i,1}^T$, with $1 \leq i \leq q$, and $x_{z}^F$, with $1 \leq z \leq q$, which are connected to $v_{C,j}$. 
It follows that $v_{C,j}$ has distance $2$ in $G[V_2]$ from $r'_T$ and from $r'_F$, 
and at most $3$ from each $x_{h,1}^T \in V_2$, with $1 \leq h \leq q$, and from each $x_{z}^F \in V_2$, with $1 \leq z \leq q$.
Furthermore, notice that, since $v_{C,j}$ is adjacent to $x_z^F$ and $x_z^F$ is adjacent to each $x_{h,2}^T \in V_2$, with $1 \leq h \leq q$ and $h \neq z$, then $v_{C,j}$ has distance at most two in $G[V_2]$ from each $x_{h,2}^T \in V_2$. 
Finally, since $v_{C,j}$ is adjacent to $y$, it has distance two and three respectively, from $y_2$ and $y_1$, in $G[V_2]$.

Consider a vertex $x_{i,1}^T \in V_2$, with $1 \leq i \leq q$. 
We have already shown that it
%
%Since it is connected
%to some vertex $v_{C,j}$, with $1 \leq j \leq p$, it 
%follows that 
has distance at most three in $G[V_2]$ from any $v_{C,j}$, with $1 \leq j \leq p$, and two from $y$.
Since $x_{i,1}^T$ is adjacent to $r'_T$, it has distance at most two from each other vertex $x_{h,1}^T$, with $1 \leq h \leq q$, and three from each other vertex $x_{h,2}^T$ of $G[V_2]$. 
Moreover, it has distance two from $y_1$ and three from $y_2$ and $r'_F$. 
Since $x_{i,2}^T$ is adjacent to every vertex $x_z^F \in V_2$, with $1 \leq z \leq q$,
as $z \neq i$,
it follows that $x_{h,1}^T$ has distance at most two from every vertex $x_z^F \in V_2$.
%and distance three from $r_F$.

Consider a vertex $x_{i,2}^T \in V_2$, with $1 \leq i \leq q$. 
We have already shown that it has distance at most two from each $v_{C,j}$
%(since $v_{C,j}$ is connected
%to at least one $x_{h}^F \in V_2$ and $x_{i,2}^T$ is adjacent to $x_h^F$) 
in $G[V_2]$. 
Since it is connected to $x_{i,1}^T$, it has distance three from $y$ and two from $r'_T$ in $G[V_2]$.
By construction $x_{i,2}^T$ is adjacent to every vertex $x_z^F \in V_2$, with $1 \leq z \leq q$,
%it follows that $x_{h,1}^T$ has distance at most two from every vertex $x_z^F \in V_2$,
$x_{i,2}^T$ has distance at most two from $r'_F$ in $G[V_2]$.
Moreover, $x_{i,2}^T$ has distance two from each vertex $x_{h,2}^T$ in $G[V_2]$, 
with $1 \leq i \leq q$, since by construction they are both adjacent to $y_1$.
%each vertex 
%$x_z^F \in V_2$, with $1 \leq z \leq q$ and we have assumed 
%that there exists at least a vertex $x_z^F \in V_2$.
Since $x_{i,2}^T$ is adjacent to $y_1$, thus it has distance at most two from $y_2$ in $G[V_2]$.

Consider a vertex $x_h^F$, with $1 \leq h \leq q$. 
It has distance one from $r'_F$ in $G[V_2]$, and thus distance two from $y_1$ and three from $y_2$ in $G[V_2]$.
Moreover, $x_h^F$ is adjacent to each $x_{i,2}^T \in V_2$, with $1 \leq i \leq q$, 
thus it has distance two from each $x_{i,1}^T$ and distance three from $r'_T$ in $G[V_2]$. 
Since by construction there exists at least one $v_{C,j}$, with $1 \leq j \leq p$, adjacent to
$x_h^F$, thus $x_h^F$ has distance two from $y$ and three from each $v_{C,z}$ in $G[V_2]$.

Finally, we consider vertices $r'_T$, $r'_F$, $y_1$ and $y_2$. 
Notice that it suffices to show that these vertices have pairwise distance at most three in $G[V_2]$, 
since we have previously shown
that any other vertex of $V_2$ has distance at most three from these vertices in $G[V_2]$. 
Since $r'_T, r'_F, y_2 \in N(y_1)$, they are all at distance at most two.
%Vertex $r'_T$ is adjacent to $y_1$, thus it has distance two from $y_2$ and $r'_F$ in $G[V_2]$. 
%Vertex $r'_F$ is adjacent to $y_1$, and it has distance two from $y_2$ in $G[V_2]$. 
%Finally, $y_1$ and $y_2$ are adjacent, 
It follows that $G[V_2]$ is a $3$-club, thus concluding the proof.
\qed\end{proof}

\begin{lemma}
\label{lem:3-club(2)2}
Given an instance $(\mathcal{C}, X)$ of $\DoubleSat$, let $G=(V,E)$ be the corresponding
instance of $\Cover{3}{2}$. 
Then, given two $3$-clubs that cover $G$, we can compute in polynomial time a truth assignment that double-satisfies $\mathcal{C}$.
\end{lemma}
\begin{proof}
Consider two $3$-clubs $G[V_1]$, $G[V_2]$, with $V_1, V_2 \subseteq V$,
that cover $G$.
First, notice that by Lemma~\ref{lem:3-club(2)Prel1} we assume that 
$r,r' \in V_1 \setminus V_2$, while $y,r'_T,r'_F \in V_2 \setminus V_1$ and 
$v_{C,j} \in V_2 \setminus V_1$, for each $j$ with $1 \leq j \leq p$.
Moreover, by Lemma \ref{lem:3-club(2)Prel} it follows that for each $i$ with $1 \leq i \leq q$, $x_{i,1}^T$ and $x_{i}^F$ do not belong to the same $3$-club, that is exactly one belongs to $V_1$ and exactly one belongs to $V_2$.

By construction, each path of length at most three from a vertex $v_{C,j}$, with $1 \leq j \leq p$, to $r'_F$ must pass through some $x_h^F$, with $1 \leq h \leq q$. 
Similarly, each path of length at most three from a vertex $v_{C,j}$,  with $1 \leq j \leq p$, to $r'_T$ must pass through some $x_{i,1}^T$.
%We start by showing the latter property. 
Assume that $v_{C,j}$,  with $1 \leq j \leq p$, is not adjacent to a vertex 
$x_{i,1}^T \in V_2$, 
with $1 \leq i \leq q$ ($x_{h}^F \in V_2$, with $1 \leq h \leq p$ respectively). 
It follows that $v_{C,j}$ is only adjacent to $y$ and to vertices $x_w^F$,  
with $1 \leq w \leq q$ ($x_{u,1}^T$,  with $1 \leq u \leq q$, respectively) in $G[V_2]$. 
In the first case, notice that $y$ is adjacent only to $v_{C,z}$,  with $1 \leq z \leq p$, 
and $y_2$, none of which is adjacent to $r'_T$ ($r'_F$, respectively), 
thus implying that this path from $v_{C,j}$ to $r'_T$ (to $r'_F$, respectively) has length at least $4$.
In the second case, $x_w^F$ ($x_{u,1}^T$, respectively) is adjacent to $r'_F$, $r_F$, $v_{C,j}$
and $x_{i,2}^T$ ($r'_T$, $r_T$, $v_{C,j}$, $x_{u,2}^T$, respectively), none of which is adjacent to
$r'_T$ ($r'_F$, respectively), implying that also in this case the path from $v_{C,j}$ to 
$r'_T$ (to $r'_F$, respectively) has length at least $4$.
Since $r'_T, r'_F , v_{C,j} \in V_2$, it follows that, for each $v_{C,j}$, the
set $V_2$ contains a vertex $x_{i,1}^T$, with $1 \leq i \leq q$, and a vertex $x_h^F$, with $1 \leq h \leq q$, connected to $v_{C,j}$.
By Lemma~\ref{lem:3-club(2)Prel} exactly one of $x_{i,1}^T$, $x_i^F$ belongs to $V_2$, 
thus we can construct a truth assignment $f$ as follows:
$f(x_i):= \text{ true}$, if $x_{i,1}^T \in V_2$,
$f(x_i):= \text{ false}$, if $x_{i}^F \in V_2$.
%
%
%\[
%f(x_i):=
%\left\{
%\begin{array}{ll}
%true 	& \mbox{if   $x_{i,1}^T \in V_2$}\\
%false 	& \mbox{if $x_{i}^F \in V_2$}
%\end{array}
%\right.
%\]
%Notice that $f$ is well-defined since by Lemma \ref{lem:3-club(2)Prel} exactly one of
%$x_{j,1}^T$, $x_{j}^F$ belongs to  $V_2$.
The assignment $f$ double-satisfies each clause of $\mathcal{C}$, since each $v_{C,j}$ is connected to a vertex $x_{i,1}^T$, for some $i$ with $1 \leq i \leq q$, and a vertex $x_{h}^F$, for some $h$ with $1 \leq h \leq q$.
%and $Y \in \{  T,F\}$, such that 
%$\{ v_{C,i}, x_{j,2}^Y  \} \in E$ and $x_{j,2}^Y \in V_2$, it follows
%that for each clause $C_i$ there exists a true literal, thus implying that $f$ satisfies $C$.
%If exactly one of $r_T$, $r_F$ belongs to $V_1$

%Moreover, notice that by construction $r_T, r_F \in V_1$, otherwise 
%$x_{i,1}^T$, $x_{j,1}^F$, for each $i$ with $1 \leq i \leq q$,  must be in $V_2$,
%as they would not be connected to $r$. 
\qed\end{proof}

Based on Lemma~\ref{lem:3-club(2)1} and Lemma~\ref{lem:3-club(2)2}, and on the
NP-completeness of $\DoubleSat$ (see Theorem \ref{teo:5DoubleSatHard}), we can conclude that 
$\Cover{3}{2}$ is NP-complete.

\begin{theorem}
\label{teo:cover3-2}
$\Cover{3}{2}$ is NP-complete.
\end{theorem}
\begin{proof}
By Lemma~\ref{lem:3-club(2)1} and Lemma~\ref{lem:3-club(2)2}, and from the NP-hardness of 
$\DoubleSat$  (see Theorem \ref{teo:5DoubleSatHard}), it follows that
$\Cover{3}{2}$ is NP-hard. 
The membership in NP follows easily from the fact that, given two $3$-clubs,
it can be checked in polynomial time whether are $3$-clubs and cover all vertices of $G$.
\qed\end{proof}

\section{Hardness of Approximation}
\label{sec:HardApprox}
\sloppy
In this section we consider the approximation complexity
of $\MinCov{2}$ and $\MinCov{3}$
and we prove that $\MinCov{2}$ is not approximable 
within factor $O(|V|^{1/2 - \varepsilon})$, for each $\varepsilon>0$, 
and that $\MinCov{3}$ is not approximable 
within factor $O(|V|^{1 - \varepsilon})$, for each $\varepsilon>0$.
%by giving a preserving-factor reduction from $\MinPart$. 
%by giving a preserving-factor reduction from $\MinColoring$. 
The proof for $\MinCov{2}$ is obtained with a reduction very similar to that of 
Section~\ref{sec:compl:2-3}, except from the fact that we reduce 
$\MinPart$ to $\MinCov{2}$.

%\todo{A nasty reviewer could say that the approximation results implies the NP-hardness results by fixing $k$ to 3...no?}

%\todo{Riccardo: yes, that's right. But on the other hard we prove the hardness of approximation in a corollary and for the organization of the paper is better to first state the complexity result. }

\begin{corollary}
\label{cor:cov2hard}
Unless $P = NP$, $\MinCov{2}$ is not approximable within factor $O(|V|^{1/2 - \varepsilon})$, for each $\varepsilon>0$.
\end{corollary}
\begin{proof}
We present a preserving-factor reduction from  $\MinPart$ to $\MinCov{2}$.
Let $G^p=(V^p,E^p)$ be a graph input of $\MinPart$, 
we compute in polynomial time a corresponding
instance $G=(V,E)$ of $\MinCov{2}$ as in Section~\ref{sec:compl:2-3}.
In what follows we prove the following results that are useful for the reduction.

%\subsection*{Proof of Lemma~\ref{lem:MinCover2HardApprox1}}

%\setcounter{lemma}{6}

\begin{lemma}
\label{appendix-lem:MinCover2HardApprox1}
Let $G^p=(V^p,E^p)$ be a graph input of $\MinPart$ and let $G=(V,E)$ be the corresponding
instance of $\MinCov{2}$. Then, given a solution of $\MinPart$ on $G^p=(V^p,E^p)$
consisting of $k$ cliques, we can compute in polynomial time a solution of $\MinCov{2}$ on $G=(V,E)$ 
consisting of $k$ $2$-clubs.
\end{lemma}
\begin{proof}
Consider a solution of $\MinPart$ on $G^p=(V^p,E^p)$ where $\{V^p_1, V^p_2, $ $\dots$ $,V^p_k\}$ is the set of $k$ cliques that partition $V^P$. 
%Then, for each $c_d$ in $C$, with $1 \leq d \leq k$,
%define %the set of colors 
%with $1 \leq d \leq k$. 
We define a solution of $\MinCov{2}$ on $G=(V,E)$ consisting of $k$ $2$-clubs as follows.
For each $d, 1 \leq d \leq k$, let 
\[
V_d = \{ w_j \in V: v_j \in V^p_{d} \} \cup \{ w_{i,j}:  v_i \in V^p_{d} \wedge i < j \}
\]
As for the proof of Lemma~\ref{lem:3-club(2)1}, it follows that for each $d$, $G[V_d]$ is a $2$-club. 
Furthermore, $G[V_1], \ldots, G[V_k]$ cover each vertex of $V$, as each $v_i \in V^p$ is covered by one of the cliques $V^p_1, V^p_2 \dots V^p_k$. 
%assigned a distinct color in $C$.
\qed\end{proof}

\begin{lemma}
\label{appendix-lem:MinCover2HardApprox2}
Let $G^p=(V^p,E^p)$ be a graph input of $\MinPart$ and let $G=(V,E)$ be the 
corresponding instance of $\MinCov{2}$. 
Then, given a solution of $\MinCov{2}$ on $G=(V,E)$ consisting of $k$ $2$-clubs, 
we can compute in polynomial time a solution of $\MinPart$ on $G^p=(V^p,E^p)$ 
with $k$ cliques.
\end{lemma}
\begin{proof}
Consider the $2$-clubs $G[V_1], \ldots, G[V_k]$ that cover $G$.
As for the proof of Lemma~\ref{lem:3-club(2)2}, 
the result follows from the fact that 
by Lemma~\ref{lem:cover2-3-Prop1},
given $w_i, w_j \in V_d$, for each $d$ with $1 \leq d \leq k$, 
it holds that $\{ v_i,v_j \} \in E$.
As a consequence, we can define a solution of $\MinPart$ on $G^p=(V^p,E^p)$ 
consisting of $k$ cliques as follows, for each $d, 1 \leq d \leq k$:
\[
V^p_d=\{v_i : w_i \in V_d \}
\]
\qed\end{proof}

The inapproximability of $\MinCov{2}$ follows 
from Lemma~\ref{appendix-lem:MinCover2HardApprox1} and 
Lemma~\ref{appendix-lem:MinCover2HardApprox2},
and from the inapproximability of $\MinPart$, which is known to be inapproximable 
within factor $O(|V^p|^{1-\varepsilon'})$~\cite{DBLP:journals/toc/Zuckerman07} 
(where $G^p=(V^p,E^p)$ is an instance of
Hence $\MinCov{2}$ is not approximable within factor $O(|V^p|^{1-\varepsilon'})$, 
for each $\varepsilon' > 0$, unless $P = NP$, hence 
$\MinCov{2}$ is not approximable within factor $O(|V^p|^{(1-\varepsilon')})$.
By the definition of $G=(V,E)$, it holds 
$|V|=|V^p|+ |E^p|  \leq |V^p|^2$
hence, for each $\varepsilon>0$,
$\MinCov{2}$ is not approximable within factor $O(|V|^{1/2 - \varepsilon})$, 
unless $P = NP$.
\qed\end{proof}

%then we present a reduction to prove the hardness of approximation of $\MinCov{3}$.
%
%\subsection{Hardness of Approximation of $\MinCov{2}$}
%\label{sec:HardApproxMinCov2} 
%
%\todo{maybe we can put all this subsection in appendix since the construction is the same as %for the np-hardness?}
%
%\todo{Riccardo: I agree, we need some space.}

\sloppy
Next, we show that $\MinCov{3}$ is not approximable within 
factor $O(|V|^{1 - \varepsilon})$,  
for each $\varepsilon>0$, unless $P = NP$, by giving a preserving-factor reduction from $\MinPart$. 

Consider an instance $G^p=(V^p,E^p)$ of $\MinPart$,
we construct an instance $G=(V,E)$ of $\MinCov{3}$ by adding a pendant
vertex connected to each vertex of $V^p$.
Formally, %$G=(V,E)$ is defined as follows: 
%(see Fig.~\ref{fig:3-clubHardApprox}):
%\[
$V = \{ u_i, w_i: v_i \in V^p  \}$, 
%\]
%We define the set of edges as follows:
%\[
$E = \{ \{u_i,w_i\}:  1 \leq i \leq |V^p| \} 
\cup \{ \{ u_i,u_j \}: \{ v_i,v_j \} \in E^p \}  \} $.
%\]

%\todo{Isn't it easier to just say we take the same graph and add a pendent vertex to %each vertex? :)}

We prove now the main properties of the reduction.

\begin{lemma}
\label{lem:MinCover3HardApprox1}
Let $G^p=(V^p,E^p)$ be an instance of $\MinPart$ and let $G=(V,E)$ be the corresponding
instance of $\MinCov{3}$. Then, given a solution of $\MinPart$ on $G^p=(V^p,E^p)$ consisting
of $k$ cliques, we can compute in polynomial time a solution of 
$\MinCov{3}$ on $G=(V,E)$ consisting of $k$ $3$-clubs.
\end{lemma}
\begin{proof}
Consider a solution of $\MinPart$ on $G^p=(V^p,E^p)$, consisting of the
cliques $\{ G^p[V_{c,1}], G^p[V_{c,2}], \dots, G^p[V_{c,k}]\}$. 
Then, for each $i$, with $1 \leq h \leq k$, define the following subset 
$V_h \subseteq V$:
\[
V_h = \{ u_j,w_j \in V: v_j \in V^p_{h} \}
\] 
Since $V^p_{1}, V^p_{2} \dots V^p_{k}$ partition $V^p$,
it follows that $V_{1}, V_{2} \dots V_{k}$ partition
(hence cover) $G$. 
Now, we show that each $G[V_h]$, with $1 \leq h \leq k$, is a $3$-club.
First, notice that since $G[V^p_{h}]$, is a clique, then 
the set $\{ u_j: u_j \in V_h \}$ induces a clique in $G$.
Then, it follows that, for each $u_i,w_j, w_z \in V_h$,
$d_{G[V_h]}(u_i,w_j) \leq 2$ and
$d_{G[V_h]}(w_j,w_z) \leq 3$, thus concluding the proof.
\qed\end{proof}

\begin{lemma}
\label{lem:MinCover3HardApprox2}
Let $G^p=(V^p,E^p)$ be a graph input of $\MinPart$ and let $G=(V,E)$ be the corresponding instance of $\MinCov{3}$. 
Then, given a solution of $\MinCov{3}$ on $G=(V,E)$ consisting
of $k$ $3$-clubs, we can compute in polynomial time a solution of $\MinPart$ 
on $G^p=(V^p,E^p)$ consisting of $k$ cliques.
\end{lemma}
\begin{proof}
Consider the $k$ $3$-clubs $G[V_1],\dots , G[V_k]$ that cover $G$. 
First, we  show that for each $V_h, 1 \leq h, \leq k$, $\forall w_i, w_j \in V_h$, 
with $1 \leq i,j \leq |V^p|$, it holds that $u_i,u_j \in V_h$.
Indeed, notice that $N(w_i)=\{u_i\}$ and $N(w_j)=\{ u_j \}$, and by the definition of a
$3$-club we must have $d_{G[v_h]}(w_i,w_j) \leq 3$, 
it follows that $ u_i,u_j  \in V_h$. 
Hence, we can define a set of cliques of $G^p$.
For each $V_h$, with $1 \leq h \leq k$, define a set $V^p_{h}$:
\[
V^p_{h}= \{ v_i:w_i \in V_h\}
\]
Notice that each $V^p_{h}$, $1 \leq h \leq k$, induces a clique in $G^p$,
as by construction if $v_i,v_j \in V^p_{h}$, then 
$w_i, w_j \in V_h$, and this implies $\{v_i,v_j\} \in E^p$.
Notice that the cliques $V^p_{1}, \dots, V^p_{k}$ may overlap, but starting from $V^p_{1}, \dots, V^p_{k}$, 
we can easily compute in polynomial time 
a clique partition of $G^p$ consisting of at most $k$ cliques.
\qed\end{proof}

Lemma~\ref{lem:MinCover3HardApprox1} and Lemma~\ref{lem:MinCover3HardApprox2}
imply the following result.

\begin{theorem}
\label{teo:FinalMinCover3HardApprox}
$\MinCov{3}$ is not approximable within factor $O(|V|^{1 - \varepsilon})$, 
for each $\varepsilon>0$, unless $P = NP$.
\end{theorem}
\begin{proof}
The result follows from Lemma~\ref{lem:MinCover3HardApprox1} and Lemma~\ref{lem:MinCover3HardApprox2},
as these results imply that we have defined a factor-preserving reduction, 
and from the inapproximability of $\MinPart$, which is known to be 
inapproximable within factor $O(|V^p|^{1-\varepsilon})$, 
for each $\varepsilon > 0$, unless $P = NP$~\cite{DBLP:journals/toc/Zuckerman07} 
(where $G^p=(V^p,E^p)$ is an instance of $\MinPart$).
Thus, $\MinCov{3}$ is not approximable within
factor $O(|V^p|^{1-\varepsilon})$, for each $\varepsilon > 0$, unless $P = NP$, and since 
it holds $|V|=2|V^p|$,  %\todo{isn't it $|V|=n+(n^2-m)$?}
$\MinCov{3}$ is not approximable within factor $O(|V|^{1 - \varepsilon})$, 
unless $P = NP$.
\qed\end{proof}

\section{An Approximation Algorithm for $\MinCov{2}$}
\label{sec:ApproxAlgo}
In this section, we present an approximation algorithm 
for $\MinCov{2}$ that 
achieves an approximation factor of 
$2|V|^{1/2}\log^{3/2} |V|$. 
Notice that, due to the result in Section \ref{sec:HardApprox}, the approximation factor
is almost tight. % (except for a $log^2 |V|$ factor).
%Moreover, notice that $\MinCov{3}$ is not approximable
%within factor $O(|V|^{1-\varepsilon})$, so there it is
%unlikely that the result can be extended to $\MinCov{3}$.
%, respectively. 
We start by describing the approximation algorithm, 
then we present the analysis of the 
approximation factor. %% for $\MinCov{2}$.
%\todo{can we have a ratio for any $s$?}
%\todo{Riccardo: don't know, I have add it in the conclusion}

%\subsection{A Greedy Approximation Algorithm}
%\label{sec:ApproxAlgoAlgo}

%In this section we describe the greedy approximation %algorithm for $\MinCov{2}$ and $\MinCov{3}$, called Club-%Cover-Approx. 

\begin{algorithm}[H]
\SetAlgoLined
\KwData{a graph $G$}
\KwResult{a cover $\mathcal{S}$ of $G$}
$V': = V$;  /* $V'$ is the set of uncovered vertices of $G$, initialized to $V$ */\\
$\mathcal{S} := \emptyset$\;
\While{$V' \neq \emptyset$}{
Let $v$ be a vertex of $V$ such that $|N[v] \cap V'|$ is maximum\;
%$\mathcal{S} := \mathcal{S} \cup N[v]$\;
Add $N[v]$ to $\mathcal{S}$\;
$V' := V' \setminus N[v]$\;
}
\caption{Club-Cover-Approx}
\label{algo:approx}
\end{algorithm}

Club-Cover-Approx is similar to the greedy approximation algorithm for 
$\mathsf{Minimum~Dominating~Set}$ and 
$\mathsf{Minimum~Set~Cover}$.
While there exists an uncovered vertex of $G$, the Club-Cover-Approx algorithm greedily defines a $2$-club 
induced by the set $N[v]$ of vertices, with $v \in V$,
such that $N[v]$ covers the maximum number of uncovered vertices (notice that some of the vertices of $N[v]$ may already be covered).
%the closed neighbourhood  of a vertex $v$ that contains 
%the maximum number of uncovered vertices (notice that some of %the vertices of $v$ may already
%be covered).
While for $\mathsf{Minimum~Dominating~Set}$ the choice 
of each iteration is optimal, here the choice is suboptimal. Notice that
indeed computing a maximum 2-club is NP-hard. 
%, so it is unlikely that the choice at each iteration
%can be done optimally. 

Clearly the algorithm returns a feasible solution for $\MinCov{2}$, as each set $N[v]$ picked by the algorithm is a $2$-club and, by construction, each vertex of 
$V$ is covered. 
Next, we show the approximation factor yielded by the Club-Cover-Approx algorithm for $\MinCov{2}$.
%yields a $O(|V|^{1/2}\log^2 |V|)$ factor.

First, consider the set $V_D$ of vertices $v \in V$ picked
by the Club-Cover-Approx algorithm, so that $N[v]$ is added to $\mathcal{S}$.
Notice that $|V_D|=|\mathcal{S}|$ and 
that $V_D$ is a dominating set of $G$,
since, at each step, the vertex $v$ picked by the algorithm dominates each vertex in $N[v]$, and each vertex in $V$ is
covered by the algorithm, so it belongs to some $N[v]$,
with $v \in V_D$.

Let $D$ be a minimum dominating set of the input graph $G$.
By the property of the greedy approximation algorithm for 
$\mathsf{Minimum~Dominating~Set}$, the set $V_D$ has the following property \cite{DBLP:journals/jcss/Johnson74a}:
\begin{equation}
|V_D| \leq |D| \log |V|
\label{eq:approx}
\end{equation}
%Consider a solution $\mathcal{S}$ returned by Club-Cover-Approx-Algorithm.
The size of a minimum dominating set in graphs
of diameter bounded by $2$ (hence $2$-clubs) 
has been considered
in \cite{DBLP:journals/jgt/DesormeauxHHY14}, where
the following result is proven. 
%\todo{I think the correct bound is $1 + \sqrt{|V_H|+\ln(|V_H|)}$, no?}

%Now, consider the following result due to ...
%%aggiungere citazione Wyatt J. Desormeaux, Teresa W. Haynes, Michael A. Henning, Anders Yeo:
%Total Domination in Graphs with Diameter 2. Journal of Graph Theory 75(1): 91-103 (2014)
%%
\begin{lemma}[\cite{DBLP:journals/jgt/DesormeauxHHY14}]
\label{lem:DominatingSet}
Let $H=(V_H,E_H)$ be a $2$-club, then $H$ has a dominating set of size  at 
most $1 + \sqrt{|V_H|+\ln(|V_H|)}$.
%$O(|V_H|^{1/2}\log |V_H|)$. 
\end{lemma}

The approximation factor $2|V|^{1/2}\log^{3/2} |V|$
for Club-Cover-Approx is obtained by combining 
Lemma \ref{lem:DominatingSet} and Equation \ref{eq:approx}.
%Let $OPT$ be the number of $2$-clubs in an optimal solution of $\MinCov{2}$.

%and is based on the fact
%that at each iteration Club-Cover-Approx covers at least 
%$O(|V'|^{\frac{5}{7}})$ of the vertices covered by the $2$-club that
%covers the maximum number of vertices in $V'$.

\begin{theorem}
\label{teo:approx}
Let $OPT$ be an optimal solution of $\MinCov{2}$, then Club-Cover-Approx returns a solution having at most $2|V|^{1/2}\log^{3/2} |V| |OPT|$ $2$-clubs.
\end{theorem}
\begin{proof}%\todo{revierwer 1 dislikes $O$ notation for numbers}
Let $D$ be a minimum dominating set of $G$ and let $OPT$ be an optimal solution of $\MinCov{2}$.
We start by proving that $|D| \leq 2|OPT| |V|^{1/2} \log^{1/2} |V|$.
For each $2$-club $G[C]$, with $C \subseteq V$, that belongs to $OPT$, by Lemma \ref{lem:DominatingSet} there exists a dominating set $D_C$ of size at most $1 + \sqrt{|C|+\ln(|C|)} \leq 2 \sqrt{|C|+\ln(|C|)}$.
%$O(|C|^{1/2} \log |C|)$. 
Since $|C| \leq |V|$, it follows that each $2$-club $G[C]$ that belongs to $OPT$ has a dominating set of size at most $2 \sqrt{|V|+\ln(|V|)}$.
%$O(|V|^{1/2} \log |V|)$. 
%Hence, $|D| \leq O(|OPT| |V|^{1/2} \log |V|)$.
Consider $D'=\bigcup_{C \in OPT} D_C$. 
It follows that $D'$ is a dominating set of $G$, since the $2$-clubs in $OPT$ covers $G$.
Since $D'$ contains $|OPT|$ sets $D_C$ and $|D_C| \leq 2\sqrt{|V|+\ln(|V|)}$, for each $G[C] \in OPT$, it follows that $|D'| \leq 2 |OPT| \sqrt{|V|+\ln(|V|)}$.
Since $D$ is a minimum dominating set, it follows that
%\[   
$|D| \leq |D'| \leq 2|OPT| (\sqrt{|V|+\ln(|V|)})$.
%\]
By Equation \ref{eq:approx}, it holds 
$|V_D| \leq 2|D| \log |V|$ 
thus 
%\[
$|V_D| \leq 2|V|^{1/2}\ln^{1/2} |V| \log |V| |OPT| 
\leq 2|V|^{1/2}\log^{3/2} |V| |OPT|$.%\todo{we mix ln and log}
%\]
%We note that Club-Cover-Approx contains exactly $|V_D|$
%$2$-clubs, thus concluding the proof.
\qed\end{proof}

%\todo[inline]{Probably Lemma~\ref{lem:DominatinSet} does %not extends to 3-clubs. Split graphs have diameter 3 and %the DS of split graph can go from 1 (1 guy of the clique %sees all the IS) to $n/2$ (matching between IS and clique)}

\section{Conclusion}
\label{sec:conclusion}
There are some interesting direction for the problem of covering a graph with $s$-clubs.
%\todo{What about structural parameterization or complexity in some specific classes of graphs?}
From the computational complexity point of view, the main open problem %for the problem
%of covering a graph with $s$-clubs
is whether $\Cover{2}{2}$ is NP-complete or is in P. Moreover, it
would be interesting to study the computational/parameterized complexity 
of the problem in specific graph classes, as done 
for {\sf Minimum Clique Partition} \cite{DBLP:journals/dam/CerioliFFMPR08,DBLP:journals/ita/CerioliFFP11,DBLP:journals/algorithmica/PirwaniS12,DBLP:journals/gc/DumitrescuP11}.
%. For example Minimum Clique Partition
%is polynomial time solvable for graphs 
%of bounded clique-width \cite{DBLP:conf/wg/EspelageGW01}.
%
%in perfect graphs or triangle-free graphs. 
%From the approximation complexity point of view, there is a 
%gap between the inapproximabily result and the approximation factor
%for both $\MinCov{2}$ and  $\MinCov{3}$. It would be interesting to reduce this gap,
%either by strengthening the %inapproximability of the problems or by %improving the approximation ratios. 
%and %Furthermore, 
%Another interesting direction is to consider the complexity and the approximability
%of $\MinCov{s}$, for value of $s$ larger than three.
%Another interesting direction is  the parameterized %approximation of the problem.

\bibliographystyle{splncs03}

\bibliography{biblio}

%\bibliographystyle{abbrv}
%\bibliography{biblio}

\end{document}